\documentclass{emsprocart}

\usepackage{graphicx,color}
\newcommand{\ie}{\emph{i.e.}}
\newcommand{\eg}{\emph{e.g.}}
\newcommand{\cf}{\emph{cf}}

\newcommand{\Real}{\mathbb{R}}
\newcommand{\Sphere}{S}
\newcommand{\Smooth}{C}
\newcommand{\s}{L}
\newcommand{\id}{I}
\newcommand{\sobi}{\mathop{W_0^{1,2}}\nolimits}

\newcommand{\Dom}{\mathop{\mathrm{Dom}}\nolimits}
\newcommand{\dist}{\mathrm{dist}}
\newcommand{\der}{\mathrm{d}}
\newcommand{\eps}{\varepsilon}

\newtheorem{Proposition}{Proposition}
\newtheorem{Theorem}{Theorem}
\theoremstyle{remark}
\newtheorem{Remark}{Remark}

\usepackage[normalem]{ulem}
\definecolor{DarkGreen}{rgb}{0,0.5,0.1} 

\newcommand\soutD{\bgroup\markoverwith
{\textcolor{DarkGreen}{\rule[.5ex]{2pt}{1pt}}}\ULon}
\newcommand\soutP{\bgroup\markoverwith
{\textcolor{blue}{\rule[.5ex]{2pt}{1pt}}}\ULon}
\newcommand{\Hm}[1]{\leavevmode{\marginpar{\tiny%
$\hbox to 0mm{\hspace*{-0.5mm}$\leftarrow$\hss}%
\vcenter{\vrule depth 0.1mm height 0.1mm width \the\marginparwidth}%
\hbox to
0mm{\hss$\rightarrow$\hspace*{-0.5mm}}$\\\relax\raggedright #1}}}

\title[Curved quantum layers]{A lower bound to the spectral threshold 
in curved quantum layers}

\author{Pedro Freitas and David Krej\v{c}i\v{r}\'{\i}k}

\contact[psfreitas@fc.ul.pt]{Pedro Freitas, 
Department of Mathematics, 
Faculty of Human Kinetics \& Group of Mathematical Physics, 
Faculdade de Ci\^{e}ncias da Universidade de Lisboa, Campo Grande, Edif\'{\i}cio C6
1749-016 Lisboa, Portugal}
\contact[krejcirik@ujf.cas.cz]{David Krej\v{c}i\v{r}\'{\i}k,
Department of Theoretical Physics,
Nuclear Physics Institute, Academy of Sciences,
25068 \v{R}e\v{z}, Czech Republic}

\begin{document}

\maketitle

\begin{center}
\emph{Dedicated to Pavel Exner on the occasion of his 70th birthday}
\end{center}

\begin{abstract}
We derive a lower bound
to the spectral threshold of the Dirichlet Laplacian
in tubular neighbourhoods of constant radius about complete surfaces.
This lower bound is given by the lowest eigenvalue
of a one-di\-men\-sional operator depending on the radius
and principal curvatures of the reference surface.
Moreover, we show that it is optimal
if the reference surface is non-negatively curved.
\end{abstract}
%



\section{Introduction}
%
In this paper we obtain a lower bound to the lowest energy of a
quantum particle confined to the space delimited by two parallel
surfaces. We assume that these surfaces represent a perfect
hard-wall boundary, in the sense that the particle wavefunction
vanishes there, and concentrate in the case where they are unbounded. 
In agreement with the paper~\cite{DEK2} 
where these structures were introduced,
we shall use the term \emph{quantum layers} for such systems.

This rather simple model is known to be remarkably successful in
describing various aspects of electronic transport in quantum
heterostructures
(we refer to the monograph~\cite{LCM} for the physical background).
One of the main questions arising within this scope
is whether or not there are geometrically induced bound states.
Indeed, some of the most
important theoretical results in the field are a number of
theorems guaranteeing the existence of such solutions
under rather simple and general physical conditions
\cite{DEK2,CEK,LL1,LL2,LL3,Lu-Rowlett_2012}
(see also \cite{Exner-Tater_2010,KL,K10,KRT,
Haag-Lampart-Teufel_2014,Dauge-Ourmieres-Bonafos-Raymond_2015,
Lampart_2015,K-Tusek_2015}
for other mathematical studies of quantum layers).

The main contribution of the present paper
is to provide a lower bound
to the ground-state energy of the bound states.
However, our results are more general in the sense
that this lower bound also applies to situations
where the lowest energy in the spectrum does not correspond
to a bound state, but rather to a scattering state;
this happens, \eg, if the layer is periodically curved.

To obtain this lower bound, we follow an idea
similar to that used by Pavel Exner and the present authors
in~\cite{EFK} to derive a lower bound
to the spectral threshold in \emph{quantum tubes},
\ie~in the case of the configuration space
being a $d$-dimensional tube about an infinite curve,
with $d \geq 2$. More precisely, there it was shown
that the lower bound is given by the lowest Dirichlet eigenvalue
in a torus determined by the geometry of the tube.
This lower bound is optimal in the sense
that it is achieved by a tube
(about a curve of constant curvature).
However, the geometry of quantum layers is more complicated
and we shall see that the optimality
is one of the main features in which the present situation
differs from that of quantum tubes.

In view of the above physical model, the Hamiltonian of a quantum layer
can be identified with the Dirichlet Laplacian
in a tubular neighbourhood of constant radius
about a complete non-compact surface $\Sigma\subset\Real^3$.
In this paper, we proceed in a greater generality
by considering compact surfaces, too.
More precisely, we assume only that
\begin{equation}
\label{H}
\begin{array}{c}
\Sigma \mbox{ is a connected complete orientable surface of class}~\Smooth^2
\mbox{ embedded in}~\Real^3 \\
\mbox{ with bounded principal curvatures } k_1 \mbox{ and } k_2.
\end{array}
\end{equation}
Then, given a positive number~$a$ satisfying
\begin{equation}\label{Ass.Basic1}
  a \, \max\{\|k_1\|_\infty,\|k_2\|_\infty\} < 1
  \,,
\end{equation}
we introduce the tubular neighbourhood
\begin{equation}\label{tube}
  \Omega := \big\{ \mathbf{x}\in\Real^3 \ |\ \dist(\mathbf{x},\Sigma)<a \big\}
\end{equation}
and denote by~$-\Delta_D^\Omega$ the Dirichlet Laplacian in
$\s^2(\Omega)$. In addition to~(\ref{Ass.Basic1}), 
we also assume that~$\Omega$ ``does not overlap itself''
(\cf~(\ref{Ass.Basic2}) below).

If~$\Sigma$ is compact, then~$\Omega$ is bounded
and a lower bound to the spectral threshold of the Laplacian
follows by means of the Faber-Krahn inequality;
\ie, $\inf\sigma(-\Delta_D^\Omega)$ is bounded from below
by the lowest Dirichlet eigenvalue of the ball of volume~$|\Omega|$
in this case.
However, we are mainly interested in the unbounded case,
where similar arguments based on the Faber-Krahn inequality may, at best,
just provide a trivial bound
and the location of $\inf\sigma(-\Delta_D^\Omega)$
becomes difficult, since we are actually
dealing with a class of \emph{quasi-cylindrical} domains
(\cf~\cite[\S49]{Glazman} or \cite[Sec.~X.6.1]{Edmunds-Evans}).
In this note we derive the following
universal lower bound:
\begin{Theorem}\label{Theorem}
Let~$\Omega$ be as above.
One has
\begin{equation}\label{bound}
  \inf\sigma(-\Delta_D^\Omega)
  \ \geq \
  \min\left\{\lambda_1(k_1^+,k_2^-),\lambda_1(k_1^-,k_2^+)\right\}
  \,,
\end{equation}
where
$
  k_i^\pm := \pm\sup(\pm k_i)
$,
$i\in\{1,2\}$,
and
\begin{equation}\label{lambda}
  \lambda_1(\kappa_1,\kappa_2)
  \ := \
  \inf_{\psi\in\sobi((-a,a))\setminus\{0\}} \ 
  \frac{\displaystyle \int_{-a}^a|\psi'(u)|^2
  \,(1-\kappa_1  u)\,(1-\kappa_2 u)\,\der u}
  {\displaystyle \int_{-a}^a|\psi(u)|^2
  \,(1-\kappa_1  u)\,(1-\kappa_2 u)\,\der u}
\end{equation}
for constants $\kappa_1, \kappa_2 \in [-1/a,1/a]$.
\end{Theorem}

In view of Theorem~\ref{Theorem},
the spectral threshold of the Dirichlet Laplacian
in the three-dimensional tubular manifold~$\Omega$
can be estimated from below by means of
the one-dimensional spectral problem
associated with~(\ref{lambda}).
It is easy to verify that $\lambda_1(k_1,k_2)$
with constant~$k_1$ and~$k_2$ gives the spectral threshold
of the Dirichlet Laplacian in the layer
about the plane if $k_1=k_2=0$,
a sphere if $k_1=k_2>0$
or a cylinder if $k_1>0$ and $k_2=0$.
That is, Theorem~\ref{Theorem} is optimal
for the class of layers built about surfaces
with non-negative Gauss curvature $k_1 k_2$.
On the other hand,
we are not aware of a geometric
meaning of~(\ref{lambda}) 
if the Gauss curvature $k_1 k_2$ is negative and the surface is complete. 
In fact, since no such surface exists which satisfies hypothesis~\eqref{H}
and whose Gauss curvature is identically equal to a negative constant,
a better lower bound than~(\ref{bound})
is expected to hold for layers about surfaces
with sign-changing or non-positive Gauss curvature.

In any case,
while the right hand side of~(\ref{bound}) diminishes
as the Gauss curvature of~$\Sigma$ becomes more negative,
it is uniformly bounded away from zero
for layers about surfaces whose Gauss curvature is non-negative:
\begin{Proposition}\label{Prop}
Let $\kappa_1,\kappa_2\in(-1/a,1/a)$ be such that $\kappa_1 \kappa_2 \geq 0$. 
Then
$$
  \lambda_1(\kappa_1,\kappa_2)
  \geq j_{0,1}^2/(2a)^2
  \,,
$$
where $j_{0,1} \approx 2.40$ denotes the first zero of the Bessel function~$J_0$.
\end{Proposition}

The bound of Proposition~\ref{Prop} is reminiscent 
of the uniform lower bound obtained in~\cite{EFK} for strips, 
\ie~a two-dimensional analogy of quantum layers,
by applying the Faber-Krahn inequality to
a sequence of Dirichlet annuli converging to a Dirichlet disk.

If $\kappa_1 \kappa_2 < 0$, it actually turns out that it is impossible to obtain 
a lower bound to $\lambda_1(\kappa_1,\kappa_2)$ 
for all $\kappa_1,\kappa_2\in(-1/a,1/a)$
that would not depend on~$\kappa_1$ and~$\kappa_2$,
as the following result shows:
\begin{Proposition}\label{Prop2}
We have
$$
  \lambda_1(-1/a,1/a) = 0
  \,.
$$
\end{Proposition}

The rest of this paper consists of one section
where we provide the proofs of Theorem~\ref{Theorem} and
Propositions~\ref{Prop} and~\ref{Prop2}.

\section{The proofs}
%
The central step in the proof of Theorem~\ref{Theorem}
is based on an idea adopted from~\cite{EFK}.
Roughly speaking, expressing the Laplacian~$-\Delta_D^\Omega$
in the natural coordinates parameterising the layer~(\ref{tube})
by means of ``longitudinal'' coordinates
of the reference surface~$\Sigma$
and a ``transverse'' coordinate
of the normal bundle of~$\Sigma$,
we neglect the contribution of the former
and the latter leads to a ``variable'' lower bound
of the type~(\ref{lambda}).
The constant lower bound
given by the right hand side of~(\ref{bound})
and the uniform lower bound of Proposition~\ref{Prop}
then follow from an analysis of the one-dimensional
spectral problem associated with~(\ref{lambda}).

We need to start with a detailed geometry of curved layers
adopted from~\cite{CEK}.
Let~$g$ be the Riemannian metric of~$\Sigma$
induced by the embedding.
The orientation of~$\Sigma$ is specified by 
a globally defined unit normal vector field $n:\Sigma\to\Sphere^2$.
For any point~$x\in\Sigma$, we introduce the Weingarten map
\begin{equation}\label{Weingarten}
  L_x : \ T_x\Sigma \to T_x\Sigma : \,
  \big\{ \xi \mapsto -\der n_x(\xi) \big\}
  \,.
\end{equation}
The principal curvatures~$k_1$ and~$k_2$ at~$x$
are defined as eigenvalues of~$L_x$ with respect to~$g(x)$.
Although these curvatures are \emph{a priori}
defined only locally on~$\Sigma$,
the Gauss curvature $K := k_1 k_2$ 
and the mean curvature $M := \frac{1}{2}(k_1+k_2)$
are globally defined continuous functions on~$\Sigma$. 

Let us introduce the mapping
\begin{equation}\label{layer}
  \mathcal{L} : \ \Sigma\times(-a,a) \to \Real^3 : \,
  \big\{(x,u) \mapsto x+n(x)\,u \big\} .
\end{equation}
Assuming~(\ref{Ass.Basic1}) and that
\begin{equation}\label{Ass.Basic2}
  \mathcal{L}
  \quad\mbox{is injective},
\end{equation}
this mapping induces a diffeomorphism
and the image $\mathcal{L}\big(\Sigma\times(-a,a)\big)$
coincides with~$\Omega$ defined by~(\ref{tube}).
In other words, $\Omega$~is a submanifold of~$\Real^3$
squeezed between two parallel surfaces at the distance~$a$ from~$\Sigma$.

Using~\eqref{layer}, we can identify~$\Omega$
with the Riemannian manifold $\Sigma\times(-a,a)$
endowed with the metric~$G$
induced by~$\mathcal{L}$.
One has
\begin{equation}\label{metric}
  G(x,u)=g(x)\circ\left(I_x-L_x\,u\right)^2 + \der u^2 \,,
\end{equation}
where~$\id_x$ denotes the identity map on~$T_x\Sigma$.
By the definition of principal curvatures,
it is easy to see that the measure
on $\Omega \simeq \big(\Real\times(-a,a),G\big)$
at a point $(x,u)$ acquires the form
\begin{equation}\label{volume}
  \der\Omega = \big(1-k_1(x)\,u\big)\big(1-k_2(x)\,u\big)
  \, \der\Sigma \, \der u \,,
\end{equation}
where~$\der\Sigma\,\der u$ stands for the product measure
on~$\Sigma\times(-a,a)$ at $(x,u)$.
Here $\der\Sigma = |g(x)|^{1/2} \der x^1 \der x^2$ 
in a local coordinate system of~$\Sigma$ at~$x$,
with the usual notation $|g| := \det(g)$.

Let~$G^{ij}$ be the coefficients of the inverse
of~$G$ in local coordinates~$(x,u)$ for~$\Sigma\times(-a,a)$.
Using the above identification,
$-\Delta_D^\Omega$ is unitarily equivalent to
the self-adjoint operator~$H$ associated with the quadratic
form~$h$ defined in the Hilbert space
$\mathcal{H} := \s^2\big(\Sigma\times(-a,a),\der\Omega\big)$ by
\begin{align}\label{form}
  h[\Psi] &:= \int_{\Sigma\times(-a,a)}
  \big(\overline{\partial_i\Psi(x,u)}\big)
  \, G^{ij}(x,u) \,
  \big(\partial_j\Psi(x,u)\big)
  \, \der\Omega
  \,,
  \\
  \Psi \in \Dom h &:=
  \sobi\big(\Sigma\times(-a,a),\der\Omega\big)
  \,. \nonumber
\end{align}
Here the Sobolev space $\sobi\big(\Sigma\times(-a,a),\der\Omega\big)$
is defined as the completion of functions from $C_0^\infty(\Sigma\times(-a,a))$
with respect to the norm 
$
  (h[\cdot] + \|\cdot\|_{\mathcal{H}}^2)^{1/2}
$.
Consequently, to prove Theorem~\ref{Theorem},
it is equivalent to establish
the lower bound~(\ref{bound}) for the operator~$H$.

\begin{proof}[Proof of Theorem~\ref{Theorem}]
Let~$\Psi$ be any function defined in
$\Smooth_0^\infty\big(\Sigma\times(-a,a)\big)$,
a dense subspace of $\Dom h$.
Since~$(G^{\mu\nu})_{\mu,\nu=1,2}$ is positive definite,
one has
\begin{align*}
  h[\Psi]
  &\geq \int_\Sigma \der\Sigma
  \int_{-a}^a \der u \ |\partial_u\Psi(x,u)|^2 \,
  \big(1-k_1(x)\,u\big)\big(1-k_2(x)\,u\big)
  \\
  &\geq \int_\Sigma \der\Sigma \ \lambda_1\big(k_1(x),k_2(x)\big)
  \int_{-a}^a \der u \ |\Psi(x,u)|^2 \,
  \big(1-k_1(x)\,u\big)\big(1-k_2(x)\,u\big)
  \,,
\end{align*}
where $\lambda_1(\kappa_1,\kappa_2)$ is defined by~(\ref{lambda}).
It remains to show that
\begin{equation}\label{remains}
  \lambda_1\big(k_1(x),k_2(x)\big)
  \geq
  \min\left\{\lambda_1(k_1^+,k_2^-),\lambda_1(k_1^-,k_2^+)\right\}
\end{equation}
for all $x \in \Sigma$.
Given constants $\kappa_1,\kappa_2 \in (-1/a,1/a)$,
the change of test function
$
  \phi := \sqrt{(1-\kappa_1 u)(1-\kappa_2 u)} \, \psi
$
in~(\ref{lambda}) and an integration by parts yields
\begin{equation}\label{lambda.potential}
  \lambda_1(\kappa_1,\kappa_2)
  \ = \
  \inf_{\phi\in\sobi((-a,a))\setminus\{0\}} \
  \frac{\displaystyle \int_{-a}^a
  \left(|\phi'(u)|^2+V(u;\kappa_1,\kappa_2)\,|\phi(u)|^2\right)
  \der u}
  {\displaystyle \int_{-a}^a|\phi(u)|^2\,\der u}
  \,,
\end{equation}
where
\begin{equation}\label{potential.pre}
  V(u;\kappa_1,\kappa_2)
  := -\frac{1}{4}\,
  \frac{(\kappa_1-\kappa_2)^2}{(1-\kappa_1 u)^2(1-\kappa_2 u)^2}
  \,.
\end{equation}
The constant lower bound~(\ref{remains}) then follows
by observing that
$$
  V\big(u;k_1(x),k_2(x)\big)
  \geq \min\big\{
  V(u;k_1^+,k_2^-), V(u;k_1^-,k_2^+)
  \big\}
$$
for any fixed $u\in(-a,a)$ and all $x \in \Sigma$.
The last inequality can be established
for non-zero~$u$'s by writing
\begin{equation}\label{potential}
  V(u;\kappa_1 ,\kappa_2 )
  = - \frac{1}{4 u^2}
  \left[\frac{1}{(1-\kappa_1  u)}-\frac{1}{(1-\kappa_2  u)}\right]^2
\end{equation}
and follows more easily for $u=0$.
\end{proof}
\begin{Remark}
Following~\cite[Rem.~1]{CEK},
since the hypothesis~(\ref{Ass.Basic1}) is still enough
to ensure that $\big(\Sigma\times(-a,a),G\big)$ is immersed in~$\Real^3$,
we do not need to assume~(\ref{Ass.Basic2})
in order to get~(\ref{bound}) for the operator~$H$.
\end{Remark}

Let us now derive the uniform lower bound
of Proposition~\ref{Prop}.
\begin{proof}[Proof of Proposition~\ref{Prop}]
In view of~\eqref{lambda}, without loss of generality
we may assume that~$\kappa_1$ and~$\kappa_2$ are non-negative.
By~\eqref{lambda.potential} with~\eqref{potential}, we have
$$
  \lambda_1(\kappa_1,\kappa_2) 
  \geq \min\{\lambda_1(\kappa_1,0),\lambda_1(0,\kappa_2)\}
  \,.
$$
However, $\lambda_1(\kappa,0)=\lambda_1(0,\kappa)$ with $\kappa \in [0,1/a)$
is the spectral threshold of the Dirichlet Laplacian in the strip
of cross-section $(-a,a)$ built either over a circle of curvature~$\kappa$
if $\kappa \not= 0$ or over a straight line if $\kappa=0$.
With help of monotonicity properties established in \cite[Prop.~4.2]{EFK}
(or using again~\eqref{lambda.potential} with~\eqref{potential}),
the Faber-Krahn inequality yields (\cf~\cite[Prop.~4.5]{EFK}) 
$$
  \lambda_1(0,\kappa) 
  \geq \lambda_1(0,1/a)
  = j_{0,1}^2/(2a)^2
$$
for all $\kappa \in [0,1/a)$.
Notice that $\lambda_1(0,1/a)$ is the lowest eigenvalue
of the Dirichlet Laplacian in the disk of radius~$2a$.
\end{proof}

Finally, we establish Proposition~\ref{Prop2}.
\begin{proof}[Proof of Proposition~\ref{Prop2}]
For any positive number $\eps < \min\{1,a\}$, let us set
$$
  \psi_\eps(u) :=
  \begin{cases}
    1 
    & \mbox{if} \quad |u| \leq a-\eps \,,
    \\
    \displaystyle
    -\frac{\log[(a-u)/\eps^2]}{\log(\eps)} 
    & \mbox{if} \quad a-\eps \leq |u| \leq a-\eps^2 \,
    \\
    0 & \mbox{if} \quad a-\eps^2 \leq |u|.
  \end{cases}
$$

Then $\psi_\eps \in W_0^{1,2}((-a,a))$ and using~$\psi_\eps$
as a test function in the right hand side of~\eqref{lambda} 
with $\kappa_1 := -1/a$ and $\kappa_2 := 1/a$, 
we obtain
$$
  \lambda_1(-1/a,1/a)
  \leq 2 \, \frac{\displaystyle\int_0^a |\psi_\eps'(u)|^2 \, (1-u/a) \, \der u}
  {\displaystyle\int_0^a |\psi_\eps(u)|^2 \, (1-u/a) \, \der u}
  \, ,
$$
where we used the bounds $1\leq 1+u/a\leq 2$.
While the denominator converges to $\int_0^a (1-u/a) \, \der u = a/2$,
an explicit computation shows that the numerator 
tends to zero as $\eps \to 0$.
\end{proof}
\begin{Remark}
Note that~\eqref{lambda.potential} yields a Hardy-Poincar\'e-type inequality
\begin{equation}\label{Hardy}
  \int_{-a}^a |\phi'(u)|^2 \, \der u
  \geq \lambda_1(\kappa_1,\kappa_2) \int_{-a}^a |\phi(u)|^2 \, \der u
  + \int_{-a}^a |V(u;\kappa_1,\kappa_2)| \, |\phi(u)|^2 \, \der u
\end{equation}
for all $\phi \in W_0^{1,2}((-a,a))$ and $\kappa_1,\kappa_2 \in [-1/a,1/a]$,
where the Hardy weight $V(\cdot;\kappa_1,\kappa_2)$ is given by~\eqref{potential.pre}
and the Poincar\'e constant $\lambda_1(\kappa_1,\kappa_2)$ interpolates 
between~$0$ and $\pi^2/(2a)^2$.
An equivalent version of this inequality in weighted spaces
follows from~\eqref{lambda}.
If $\kappa_1=\kappa_2 \geq 0$, then $V(\cdot;\kappa_1,\kappa_2)$
vanishes identically and $\lambda_1(\kappa_1,\kappa_2)$ equals $\pi^2/(2a)^2$,
the first eigenvalue of the Dirichlet Laplacian in the interval $(-a,a)$.
On the other hand, putting $\kappa_1=1/a$ and $\kappa_2=-1/a$ in~\eqref{Hardy},
Proposition~\eqref{Prop2} yields an optimal Hardy-type inequality
\begin{equation}
  \int_{-a}^a |\phi'(u)|^2 \, \der u
  \geq \int_{-a}^a \frac{a^2}{(a^2-u^2)^2} \, |\phi(u)|^2 \, \der u
\end{equation}
for all $\phi \in W_0^{1,2}((-a,a))$.
We remark that this inequality is better 
than the well-known bound (see, \eg, \cite{BM})
$$
  \int_{-a}^a |\phi'(u)|^2 \, \der u
  \geq \int_{-a}^a \frac{1}{4 \,(a-|u|)^2} \, |\phi(u)|^2 \, \der u
$$
for all $\phi \in W_0^{1,2}((-a,a))$,
which can be established by the classical Hardy inequality.
Notice that the function $a-|\cdot|$ has the meaning of the distance
to the boundary of the one-dimensional domain~$(-a,a)$.
Hardy inequalities with weights of type~\eqref{potential.pre} 
have been recently considered for higher-dimensional domains in~\cite{Cazacu}
(see also \cite[Lem.~8]{BDE}). 
\end{Remark}
%

\section*{Acknowledgement}
%
The research of the second author was supported by 
the project RVO61389005 and the GACR grant No.\ 14-06818S.

%
%

\providecommand{\bysame}{\leavevmode\hbox to3em{\hrulefill}\thinspace}
\providecommand{\MR}{\relax\ifhmode\unskip\space\fi MR }
\providecommand{\MRhref}[2]{%
  \href{http://www.ams.org/mathscinet-getitem?mr=#1}{#2}
}
\providecommand{\href}[2]{#2}

\end{document}